\documentclass[journal,twocolumn]{IEEEtran}

\usepackage{amsthm}
\newtheorem{theorem}{Theorem}

\newcommand{\be}{\begin{equation}}
\newcommand{\ee}{\end{equation}}
\newcommand{\bea}{\begin{eqnarray}}
\newcommand{\eea}{\end{eqnarray}}
\newcommand{\bw}{\begin{eqnarray*}}
\newcommand{\ew}{\end{eqnarray*}}

\usepackage{graphicx,epstopdf}
\usepackage{amssymb,amsmath}
\usepackage{amsmath}
\usepackage{latexsym}
\usepackage{graphicx}
\usepackage{graphics}
\usepackage{color}
\usepackage{cite}

\title{Tail Index for a Distributed Storage System with Pareto File Size Distribution}
\begin{document}
\author{Vaneet Aggarwal and Tian Lan\thanks{V. Aggarwal is with the School of IE, Purdue University, West Lafayette, IN 47907 (email: vaneet@purdue.edu). T. Lan is with the Department of ECE, George Washington University, DC 20052 (email: tlan@gwu.edu).}}
\date{\vspace{-5ex}}

\maketitle

\begin{abstract}

Distributed storage systems often employ erasure codes to achieve high data reliability while attaining space efficiency. Such storage systems are known to be susceptible to long tails in response time. It has been shown that in modern online applications such as Bing, Facebook, and Amazon, the long tail of latency is of particular concern, with $99.9$th percentile response times that are orders of magnitude worse than the mean. Taming tail latency is very challenging in erasure-coded storage systems since quantify tail latency (i.e., $x$th-percentile latency for arbitrary $x\in[0,1]$) has been a long-standing open problem. In this paper, we propose a mathematical model to quantify {\em tail index} of service latency for arbitrary erasure-coded storage systems, by characterizing the asymptotic behavior of latency distribution tails. When file size has a heavy tailed distribution, we find tail index, defined as the exponent at which latency tail probability diminishes to zero, in closed-form, and further show that a family of probabilistic scheduling algorithms are (asymptotically) optimal since they are able to achieve the exact tail index.

\end{abstract}

\section{Introduction}

Modern storage systems, such as those developed by Facebook \cite{Sathiamoorthy13}, Microsoft \cite{Asure14} and Google \cite{Fikes10}, increasingly adopt erasure coding to achieve more efficient use of storage capacity while maintaining high reliability guarantees. These distributed storage systems are known to be susceptible to long tails in response time. It has been shown that in modern Web applications such as Bing, Facebook, and Amazon's retail platform, the long tail of latency is of particular concern, with $99.9$th percentile response times that are orders of magnitude worse than the mean \cite{T1,T2}. Despite mechanisms such as load-balancing and resource management, still evaluations of large scale systems indicate that there is a high degree of randomness in delay performance \cite{Liang:2014}. For distributed storage systems that use erasure coding, quantifying tail latency (i.e., $x$th-percentile latency for arbitrary $x\in[0,1]$) has been a long-standing open problem. 

This paper proposes an analytical framework to quantify {\em tail index} of service latency for arbitrary erasure-coded storage systems, by characterizing the asymptotic behavior of latency distribution tails. Existing work has mainly focused on providing bounds for mean service latency \cite{ISIT:12,Joshi:13,MDS-Queue,Xiang:2014:Sigmetrics:2014,Yu_TON,Yu-TNSM16,Yu-TCC16,Yu-TON16,Sprout,Yu-ICDCS,Yu-CCGRID} and fails to address the issue of latency tails. However, the overall response time in erasure coded data-storage systems is dominated by the long tail distribution of the parallel operations \cite{Barroso:2011,Barroso:20112,Tail_TON,Jingxian}. Evaluation of practical systems show that the latency spread is significant even when data object sizes are in the order of megabytes \cite{Liang:2014}. The lack of mathematical models prevents design and optimization of erasure-coded storage with the goal of keeping latency consistently low (rather than just reducing the mean) and meeting customer delay Service Level Agreements (SLAs).

Using an $(n,k)$ erasure code, a file is encoded into $n$ data chunks, allowing reconstruction from any subset of $k<n$ distinct chunks. Quantifying service latency for heavy-tailed files is an open problem because of the challenge of jointly analyzing dynamic scheduling (i.e., $n$-choose-$k$ optimization for each request on the fly and {\em w.r.t.} data locality and network status) and the dependency of chunk access time on shared storage nodes. For mean latency analysis of homogeneous files, {\em Fork-join queue analysis} in \cite{Makowski:89, JK13, Joshi:13, KumarTC14} provides upper bounds for mean service latency by forking each file request to all storage nodes and removing it after enough chunks are processed. In a separate line of work, {\em Queuing-theoretic analysis} in \cite{ISIT:12,MG1:12,MDS-Queue} proposes a {\em block-$t$-scheduling} policy that only allows the first $t$ requests at the head of the buffer to move forward, and finds an upper bound of the mean latency through Markov-chain analysis of the queuing model.  However, both approaches fall short of quantifying tail latency, because states of the corresponding queuing model must encapsulate not only a snapshot of the current system including chunk placement and queued requests, but also past history of how chunk requests have been processed by individual nodes. This leads to a {\em state explosion} problem as practical storage systems usually handle a large number of files and nodes. Later, mean latency bounds for arbitrary service time distribution and file configurations are provided in \cite{Xiang:2014:Sigmetrics:2014,Yu_TON} using order statistic analysis and probabilistic request scheduling policy. Recent work has also given approaches to understand the tail latency using probabilistic request scheduling policy \cite{Jingxian,Tail_TON}.


In this paper, we focus on analyzing the asymptotic behavior of latency tail distributions in erasure-coded storage systems, when file size has a heavy tail distribution motivated by heavy-tailed file size distribution in local file systems and in the World Wide Web \cite{948888,Vaphase,Gong01onthe}. To quantify tail latency, we make use of probabilistic scheduling develop in \cite{Xiang:2014:Sigmetrics:2014,Yu_TON}. Upon the arrival of each file request, we randomly dispatch a batch of $k$ chunk requests to $k$-out-of-$n$ storage nodes with some predetermined probabilities. Then, each storage node manages its local queue independently and continues processing requests in order. A file request is completed if all its chunk requests exit the system. This probabilistic scheduling policy allows us to analyze the (marginal) queuing delay distribution of each storage node and then combine the results to obtain a lower bound on the asymptotic  tail latency of any distributed storage system. 

In particular, we consider tail index, defined as the exponent at which latency tail probability diminishes to zero, i.e., $-\log {\rm Pr}(L\ge x)/ \log(x)$ as threshold $x$ grows large. When file size follows a Pareto distribution  and unit service time follows an exponential distribution, we employ Laplace-Stieltjes transform to solve in closed form the latency tail probability for processing a chunk request at any single server. Utilizing this result, we prove that tail index of erasure-coded storage systems is upper bounded by $alpha-1$, where $\alpha$ is the exponent of Pareto-distributed file size. We further show that this upper bound is indeed achievable via probabilistic scheduling policy. In order words, a family of probabilistic scheduling algorithms achieves the best tail index and are optimal with respect to asymptotic latency tails. 

The main contributions of this paper are summarized as follow:
\begin{itemize}
\item We propose an analytical framework to quantify tail index of service latency for arbitrary erasure-coded storage systems.
\item For Pareto-distributed file size (with shape parameter $\alpha>2$ ) and exponential service time, we prove that the optimal tail index of  erasure-coded storage systems is $\alpha-1$.
\item We show that a family of probabilistic scheduling algorithms are able to achieve the tail index and therefore are asymptotically optimal in terms of latency tails.
\end{itemize}

The rest of the paper is organized as follows. Section II gives the system model for the problem. Section III finds  the tail index of the waiting time from each server. Section IV uses the tail index from each server to show that the optimal tail index of distributed storage system is achieved by probabilistic scheduling. Finally, Section V concludes this paper. 
\section{System Model and Formulation}

We consider a data center consisting of $n$ heterogeneous storage
servers, denoted by ${\cal N} = \{1, 2, \cdots , n\}$, called storage nodes in this paper. There are $r$ files stored in the data-center using an $(n,k)$ erasure code. That is, each file is split into $k$ equal-size chunks and then encoded into $n$ coded chunks placed on different storage nodes. Thus, any file can be accessed by retrieving $k$ distinct data chunks. 

The arrival of client requests for each file $i$ of size $k L_i$ Mb is assumed to form an independent Poisson process with a known
rate $\lambda_i$. We assume that the chunk size $L_i$ Mb has a heavy tail and follows a Pareto distribution with parameters $(x_m,\alpha)$ with shape parameter $\alpha>2$  (implying finite mean and variance). Thus, the complementary cumulative distribution function (c.d.f.) of the chunk size is given as
\begin{equation}
\Pr(L_i>x)=\begin{cases}
(x_m/x)^\alpha \quad x\ge x_m\\
0 \quad x<x_m
\end{cases}
\end{equation}
For $\alpha>1$, the mean is $E[L_i] = \alpha x_m/(\alpha-1)$. The service time per Mb at server $j$, $X_j$ is distributed as an exponential distribution the mean service time  $1/\mu_j$.

We will focus on the tail index of the waiting time to access each file. In order to understand the tail index, let the waiting time for the files $T_W$ has $\Pr(T_W)>x)$ of the order of $x^{-d}$ for large $x$, then the tail index is $d$. More formally, the tail index $d$ is defined as $\lim_{x\to\infty}\frac{-\log \Pr(T_W>x)}{\log x}$. This index gives the slope of the tail in the log-log scale of the complementary c.d.f. 

To prove achievability of the tail index, we will consider a class of probabilistic scheduling policies \cite{Yu_TON}, which are employed in mean latency analysis and shown to provide a tight bound for arbitrary erasure code and service time distribution. In particular, under probabilistic scheduling, upon the arrival of a file $i$ request,  we randomly dispatch the batch of $k$ chunk requests to $k$ out of $n$ storage nodes in ${\cal N}$, denoted by a subset $\mathcal{A}_i\subseteq {\cal N}$ (satisfying  $|\mathcal{A}_i|=k$) with predetermined probabilities. Then, each storage node manages its local queue independently and continues processing requests in order. A file request is completed if all its chunk requests exit the system. It was further shown in \cite{Yu_TON} that a probabilistic scheduling policy with feasible probabilities $\{\mathbb{P}(\mathcal{A}_i): \ \forall i,\mathcal{A}_i\}$ exists if and only if there exists conditional probabilities $\{\pi_{i,j}\in [0,1], \forall i,j\}$ satisfying
\begin{eqnarray}
\sum_{j=1}^n \pi_{i,j} = k \ \forall i.
\end{eqnarray}

These probabilities $\pi_{i,j}$ indicate the probability with with file $i$ is requested from node $j$. It is easy to verify that under our model, the arrival of chunk requests at node $j$ form a Poisson Process with rate $\Lambda_j=\sum_i \lambda_i\pi_{i,j}$, which is the superposition of $n$ Poisson processes each with rate $\lambda_i\pi_{i,j}$. In order to show that probabilistic scheduling achieves the optimal tail index in this paper, we will only use $\pi_{ij}= k/n$ which denotes uniform access of files from all the $n$-choose-$k$ subsets. We assume that this value of $\pi_{ij}$ will lead to stable queues, such that arrival rate $\Lambda_j$ at each server is smaller than the service rate $\mu_j$ at each server.


\section{Waiting Time Distribution for a Chunk from a Server}

In this Section, we will characterize the Laplace Stieltjes transform of the waiting time distribution from a server, assuming that the arrival of requests at a server is Poisson distributed with mean arrival rate $\Lambda_j$.  We first note that the service time per chunk on server $j$ is given as $B_j = X_j L_i$, where $L_i$ is distributed as Pareto Distribution given above, and $X_j$ is exponential with parameter $\mu_j$.

Using this, we find that
\begin{eqnarray}
&&\Pr(B_j<y) \nonumber\\&=& \Pr(X_j L_i <y)\nonumber\\
&=& \int_{x=x_m}^\infty\Pr(X_j<y/x)\alpha x_m^\alpha\frac{1}{x^{\alpha+1}}  dx\nonumber\\
&=& \int_{x=x_m}^\infty\left(1-\exp(-\mu_jy/x)\right)\alpha x_m^\alpha\frac{1}{x^{\alpha+1}}  dx\nonumber\\
&=& 1 - \int_{x=x_m}^\infty \exp(-\mu_jy/x)\alpha x_m^\alpha\frac{1}{x^{\alpha+1}}  dx
\end{eqnarray}
Substitute $t = \mu_j y/x$, and then $dt = -\mu_j y/x^2 dx$. Thus,
\begin{eqnarray}
&&\Pr(B_j>y) \nonumber\\&=&  \int_{x=x_m}^\infty \exp(-\mu_jy/x)\alpha x_m^\alpha\frac{1}{x^{\alpha+1}}  dx\nonumber\\
&=&  \int_{t=0}^{\mu_j y/x_m} \exp(-t)\alpha x_m^\alpha \frac{t^{\alpha-1}}{(\mu_j y)^\alpha} dt\nonumber\\
&=&  \alpha (x_m/\mu_j)^\alpha\frac{1}{ y^\alpha} \int_{t=0}^{\mu_j y/x_m} \exp(-t) {t^{\alpha-1}}  dt\nonumber\\
&=&  \alpha (x_m/\mu_j)^\alpha \gamma(\alpha,\mu_jy/x_m)/y^\alpha,
\end{eqnarray}
 where $\gamma$ denote lower incomplete gamma function, given as $\gamma(a,x)=\int_0^x u^{a-1}\exp(-u) du$.
 
 Since $\Pr(B_j>y) = L(y)/y^\alpha$, where $L(y) = \alpha (x_m/\mu_j)^\alpha \gamma(\alpha,\mu_jy/x_m)$ is a slowly varying function, the asymptotic of the waiting time in heavy-tailed limit can be calculated using the results in \cite{olvera2011transition} as
 
 \begin{equation}
 \Pr(W>x) \approx \frac{\Lambda}{1-\rho} \frac{x^{1-\alpha}}{\alpha-1} L(x).
 \end{equation}

Thus, we note that the waiting time from a server is heavy-tailed with tail-index $\alpha-1$. Thus, we get the following result. 
\begin{theorem}
	Assume that the arrival rate for requests is Poisson distributed, service time distribution is exponential and the chunk size distribution is Pareto with shape parameter $\alpha$. Then,  the tail index for the waiting time of chunk in the queue of a server  is $\alpha-1$.
\end{theorem}

\section{Probabilistic Scheduling Achieves Optimal Tail Index}

Having characterized the tail index of a single server with Poisson arrival process and Pareto distributed file size, we will now give the tail index for a general distributed storage system. The first result is that any distributed storage system has a tail index of at most $\alpha-1$.

\begin{theorem}
	The tail index for distributed storage system is at most $\alpha-1$. 
\end{theorem}
\begin{proof}
In order to show this result,  consider a genie server which is combination of  all the $n$ servers together. The service rate of this server is $\sum_{j=1}^n \mu_i$ per Mb. As a genie, we also assume that only one chunk is enough to be served. In this case, the problem reduces to the single server problem with Poisson arrival process and the result in Section IIII shows that the tail index is $\alpha-1$. Since even in the genie-aided case, the tail index is $\alpha-1$, we cannot get any higher tail index. 
\end{proof}

The next result shows that the probabilistic scheduling achieves the optimal tail index.

\begin{theorem}
The optimal tail index of $\alpha-1$  is achieved by probabilistic scheduling. 
\end{theorem}
\begin{proof}
In order to show that probabilistic scheduling achieves this tail index, we consider the simple case where all the $n$-choose-$k$ sets are chosen equally likely for each file. Using this, we note that each server is accessed with equal probability of $\pi_{ij}=k/n$. Thus, the arrival rate at the server is Poisson and the tail index of the waiting time at the server is $\alpha-1$.

The overall latency  of a file chunk is the sum of the queue waiting time and the service time. Since the service time has tail index of $\alpha$, the overall latency for a chunk is $\alpha-1$.  Probability that latency is greater than $x$ is determined by the $k^{th}$ chunk to be received. The probability is upper bounded by the sum of probability over all servers that waiting time at a server is greater than $x$. This is because $\Pr(\max_j(A_j)\ge x)\le \sum_j \Pr(A_j\ge x)$ even when the random variables $A_j$ are correlated.  Finite sum of terms, each will tail index $\alpha-1$ will still give the term with tail index $\alpha-1$ thus proving that the tail index with probabilistic scheduling is $\alpha-1$. 
\end{proof}

We note that even though we assumed a total of $n$ servers, and the erasure code being the same, the above can be extended to the case when there are more than $n$ servers with uniform placement of files and each file using different erasure code. The upper bound argument does not change as long as number of servers are finite. For the achievability with probabilistic scheduling, we require that the chunks that are serviced follow a Pareto distribution with shape parameter $\alpha$. Thus, as long as placed files on each server are placed with the same distribution and the access pattern does not change the nature of distribution of accessed chunks from a server, the result holds in general.
\section{Conclusions}

In this paper, we prove that the tail index of arbitrary erasure-coded storage systems is $\alpha-1$, where the  file size follows a Pareto distribution with exponent $\alpha$ and the service time follows an exponential distribution. To the best of our knowledge, this is the first mathematical framework to analyze and quantify tail latency for arbitrary erasure-coded storage systems. Furthermore, we show that a family of probabilistic scheduling algorithms are optimal for tail latency in the sense that they are able to achieve the exact tail index. The results in this paper illuminate key design issues for taming tail latency in distributed storage systems that employ erasure coding.

\bibliographystyle{IEEEtran}
\bibliography{allstorage,Tian,ref_Tian2,ref_Tian3,Vaneet_cloud,Tian_rest}

\begin{thebibliography}{10}
\providecommand{\url}[1]{#1}
\csname url@samestyle\endcsname
\providecommand{\newblock}{\relax}
\providecommand{\bibinfo}[2]{#2}
\providecommand{\BIBentrySTDinterwordspacing}{\spaceskip=0pt\relax}
\providecommand{\BIBentryALTinterwordstretchfactor}{4}
\providecommand{\BIBentryALTinterwordspacing}{\spaceskip=\fontdimen2\font plus
\BIBentryALTinterwordstretchfactor\fontdimen3\font minus
  \fontdimen4\font\relax}
\providecommand{\BIBforeignlanguage}[2]{{%
\expandafter\ifx\csname l@#1\endcsname\relax
\typeout{** WARNING: IEEEtran.bst: No hyphenation pattern has been}%
\typeout{** loaded for the language `#1'. Using the pattern for}%
\typeout{** the default language instead.}%
\else
\language=\csname l@#1\endcsname
\fi
#2}}
\providecommand{\BIBdecl}{\relax}
\BIBdecl

\bibitem{Sathiamoorthy13}
M.~Sathiamoorthy, M.~Asteris, D.~Papailiopoulos, A.~G. Dimakis, R.~Vadali,
  S.~Chen, and D.~Borthakur, ``Xoring elephants: Novel erasure codes for big
  data,'' in \emph{Proceedings of the 39th international conference on Very
  Large Data Bases.}, 2013.

\bibitem{Asure14}
C.~Huang, H.~Simitci, Y.~Xu, A.~Ogus, B.~Calder, P.~Gopalan, J.~Li, and
  S.~Yekhanin, ``Erasure coding in windows azure storage,'' in
  \emph{Proceedings of the 2012 USENIX Conference on Annual Technical
  Conference}, ser. USENIX ATC'12.\hskip 1em plus 0.5em minus 0.4em\relax
  USENIX Association, 2012.

\bibitem{Fikes10}
A.~Fikes, ``Storage architecture and challenges (talk at the google faculty
  summit),'' http://bit.ly/nUylRW, Tech. Rep., 2010.

\bibitem{T1}
J.~Dean and L.~A. Barroso, ``The tail at scale,'' in \emph{Communications of
  the ACM}, 2013.

\bibitem{T2}
B.~N. M.~B. Yunjing~Xu, Zachary~Musgrave, ``Bobtail: Avoiding long tails in the
  cloud,'' in \emph{10th USENIX Symposium on Networked Systems Design and
  Implementation (NSDI ’13)}, 2013.

\bibitem{Liang:2014}
\BIBentryALTinterwordspacing
G.~Liang and U.~C. Kozat, ``Fast cloud: Pushing the envelope on delay
  performance of cloud storage with coding,'' \emph{IEEE/ACM Trans. Netw.},
  vol.~22, no.~6, pp. 2012--2025, Dec. 2014. [Online]. Available:
  \url{http://dx.doi.org/10.1109/TNET.2013.2289382}
\BIBentrySTDinterwordspacing

\bibitem{ISIT:12}
L.~Huang, S.~Pawar, H.~Zhang, and K.~Ramchandran, ``Codes can reduce queueing
  delay in data centers,'' in \emph{Information Theory Proceedings (ISIT), 2012
  IEEE International Symposium on}, July 2012, pp. 2766--2770.

\bibitem{Joshi:13}
G.~Joshi, Y.~Liu, and E.~Soljanin, ``On the delay-storage trade-off in content
  download from coded distributed storage systems,'' \emph{Selected Areas in
  Communications, IEEE Journal on}, vol.~32, no.~5, pp. 989--997, May 2014.

\bibitem{MDS-Queue}
N.~Shah, K.~Lee, and K.~Ramachandran, ``The mds queue: analyzing latency
  performance of codes and redundant requests,'' \emph{arXiv:1211.5405}, Nov
  2012.

\bibitem{Xiang:2014:Sigmetrics:2014}
\BIBentryALTinterwordspacing
Y.~Xiang, T.~Lan, V.~Aggarwal, and Y.~F.~R. Chen, ``Joint latency and cost
  optimization for erasure-coded data center storage,'' \emph{SIGMETRICS
  Perform. Eval. Rev.}, vol.~42, no.~2, pp. 3--14, Sep. 2014. [Online].
  Available: \url{http://doi.acm.org/10.1145/2667522.2667524}
\BIBentrySTDinterwordspacing

\bibitem{Yu_TON}
------, ``Joint latency and cost optimization for erasure-coded data center
  storage,'' \emph{IEEE/ACM Transactions on Networking}, vol.~24, no.~4, pp.
  2443--2457, Aug 2016.

\bibitem{Yu-TNSM16}
Y.~Xiang, T.~Lan, V.~Aggarwal, and Y.~F. Chen, ``Optimizing differentiated
  latency in multi-tenant, erasure-coded storage,'' \emph{IEEE Transactions on
  Network and Service Management}, vol.~14, no.~1, pp. 204--216, March 2017.

\bibitem{Yu-TCC16}
Y.~Xiang, V.~Aggarwal, Y.~F. Chen, and T.~Lan, ``Differentiated latency in data
  center networks with erasure coded files through traffic engineering,''
  \emph{IEEE Transactions on Cloud Computing}, vol.~PP, no.~99, pp. 1--1, 2017.

\bibitem{Yu-TON16}
\BIBentryALTinterwordspacing
V.~Aggarwal, Y.~R. Chen, T.~Lan, and Y.~Xiang, ``Sprout: {A} functional caching
  approach to minimize service latency in erasure-coded storage,'' \emph{CoRR},
  vol. abs/1609.09827, 2016. [Online]. Available:
  \url{http://arxiv.org/abs/1609.09827}
\BIBentrySTDinterwordspacing

\bibitem{Sprout}
V.~Aggarwal, Y.-F. Chen, T.~Lan, and Y.~Xiang, ``Sprout: A functional caching
  approach to minimize service latency in erasure-coded storage,'' in
  \emph{Distributed Computing Systems (ICDCS), 2016 IEEE 36th International
  Conference on}, June 2016.

\bibitem{Yu-ICDCS}
Y.~Xiang, T.~Lan, V.~Aggarwal, and Y.-F. Chen, ``Multi-tenant latency
  optimization in erasure-coded storage with differentiated services,'' in
  \emph{Distributed Computing Systems (ICDCS), 2015 IEEE 35th International
  Conference on}, June 2015, pp. 790--791.

\bibitem{Yu-CCGRID}
Y.~Xiang, V.~Aggarwal, Y.-F. Chen, and T.~Lan, ``Taming latency in data center
  networking with erasure coded files,'' in \emph{Cluster, Cloud and Grid
  Computing (CCGrid), 2015 15th IEEE/ACM International Symposium on}, May 2015,
  pp. 241--250.

\bibitem{Barroso:2011}
\BIBentryALTinterwordspacing
L.~A. Barroso, ``Warehouse-scale computing: Entering the teenage decade,'' in
  \emph{Proceedings of the 38th Annual International Symposium on Computer
  Architecture}, ser. ISCA '11.\hskip 1em plus 0.5em minus 0.4em\relax New
  York, NY, USA: ACM, 2011, pp.~--. [Online]. Available:
  \url{http://dl.acm.org/citation.cfm?id=2000064.2019527}
\BIBentrySTDinterwordspacing

\bibitem{Barroso:20112}
\BIBentryALTinterwordspacing
------, ``Warehouse-scale computing: Entering the teenage decade,''
  \emph{SIGARCH Comput. Archit. News}, vol.~39, no.~3, pp.~--, Jun. 2011.
  [Online]. Available: \url{http://doi.acm.org/10.1145/2024723.2019527}
\BIBentrySTDinterwordspacing

\bibitem{Tail_TON}
\BIBentryALTinterwordspacing
V.~Aggarwal, A.~O. Al{-}Abbasi, J.~Fan, and T.~Lan, ``Taming tail latency for
  erasure-coded, distributed storage systems,'' \emph{CoRR}, vol.
  abs/1703.08337, 2017. [Online]. Available:
  \url{http://arxiv.org/abs/1703.08337}
\BIBentrySTDinterwordspacing

\bibitem{Jingxian}
V.~Aggarwal, J.~Fan, and T.~Lan, ``Taming tail latency for erasure-coded,
  distributed storage systems,'' in \emph{Proc. IEEE Infocom}, Jul 2017.

\bibitem{Makowski:89}
A.~F.Baccelli and A.Shwartz, ``The fork-join queue and related systems with
  synchronization constraints: stochastic ordering and computable bounds,''
  \emph{Advances in Applied Probability}, p. 629–660, 1989.

\bibitem{JK13}
G.~Liang and U.~Kozat, ``Fast cloud: Pushing the envelope on delay performance
  of cloud storage with coding,'' \emph{Networking, IEEE/ACM Transactions on},
  vol.~22, no.~6, pp. 2012--2025, Dec 2014.

\bibitem{KumarTC14}
\BIBentryALTinterwordspacing
A.~Kumar, R.~Tandon, and T.~C. Clancy, ``On the latency of erasure-coded cloud
  storage systems,'' \emph{CoRR}, vol. abs/1405.2833, 2014. [Online].
  Available: \url{http://arxiv.org/abs/1405.2833}
\BIBentrySTDinterwordspacing

\bibitem{MG1:12}
L.~Huang, S.~Pawar, H.~Zhang, and K.~Ramchandran, ``Codes can reduce queueing
  delay in data centers,'' in \emph{Information Theory Proceedings (ISIT), 2012
  IEEE International Symposium on}, July 2012, pp. 2766--2770.

\bibitem{948888}
A.~B. Downey, ``The structural cause of file size distributions,'' in
  \emph{Modeling, Analysis and Simulation of Computer and Telecommunication
  Systems, 2001. Proceedings. Ninth International Symposium on}, 2001, pp.
  361--370.

\bibitem{Vaphase}
V.~Ramaswami, K.~Jain, R.~Jana, and V.~Aggarwal,
  ``\BIBforeignlanguage{English}{Modeling heavy tails in traffic sources for
  network performance evaluation},'' in
  \emph{\BIBforeignlanguage{English}{Computational Intelligence, Cyber Security
  and Computational Models}}, ser. Advances in Intelligent Systems and
  Computing.\hskip 1em plus 0.5em minus 0.4em\relax Springer India, 2014, vol.
  246, pp. 23--44.

\bibitem{Gong01onthe}
W.~Gong, Y.~Liu, V.~Misra, and D.~Towsley, ``On the tails of web file size
  distributions,'' in \emph{in: Proceedings of 39th Allerton Conference on
  Communication, Control, and Computing}, 2001.

\bibitem{olvera2011transition}
M.~Olvera-Cravioto, J.~Blanchet, P.~Glynn \emph{et~al.}, ``On the transition
  from heavy traffic to heavy tails for the m/g/1 queue: the regularly varying
  case,'' \emph{The Annals of Applied Probability}, vol.~21, no.~2, pp.
  645--668, 2011.

\end{thebibliography}

\end{document}